\documentclass[letter,aps,prl,showpacs,twocolumn,superscriptaddress]{revtex4-2}   
   
\usepackage[utf8]{inputenc}  
\usepackage[T1]{fontenc}     
\usepackage[british]{babel}  
\usepackage{lmodern}  

\usepackage[scaled=1.03]{inconsolata} 
\usepackage[usenames,dvipsnames]{color} 
\usepackage[colorlinks,citecolor=blue,linkcolor=magenta,urlcolor=blue]{hyperref}  
\usepackage{graphicx} 
\usepackage{tikz}
\usepackage[babel]{microtype}  
\usepackage{amsmath,amssymb,amsthm,bm,mathtools,amsfonts,mathrsfs,bbm,dsfont} 
\usepackage{xspace}  
\usepackage{multirow}
\usepackage{verbatim}
\usepackage{enumerate}
\usepackage{tcolorbox}
\usepackage[normalem]{ulem}

\usepackage{footmisc}

\usepackage{physics}
\newcommand{\id}{\ensuremath{\mathds{1}}}
\usepackage{bbold}
\usepackage[mathscr]{eucal}
\usepackage{units}

\newcommand{\kb}[2]{|#1\rangle\langle#2|} 
\newcommand{\Inst}{\mathcal{I}}

\newtheoremstyle{mystyle}
  {6pt}
  {6pt}
  {\normalfont}
  {0pt}
  {\bf}
  {.}
  { }
  {}

\theoremstyle{mystyle}

\newtheorem*{lemma*}{Lemma}

\newtheorem*{corollary*}{Corollary}

\newcommand{\A}{\mathsf{A}}
\newcommand{\B}{\mathsf{B}}
\newcommand{\C}{\mathsf{C}^U}

\newcommand{\W}{\mathsf{W}}
\newcommand{\K}{\mathsf{K}}

\newcommand{\obsA}{\mathcal{A}}
\newcommand{\obsB}{\mathcal{B}}
\newcommand{\obsW}{\mathcal{W}}

\begin{document}
\author{Konstantin Beyer}
\affiliation{Institut f\"ur Theoretische Physik, Technische Universit\"at Dresden, D-01062, Dresden, Germany}
\author{Roope Uola}
\affiliation{Department of Applied Physics, University of Geneva, 1211 Geneva, Switzerland}
\author{Kimmo Luoma}
\affiliation{Institut f\"ur Theoretische Physik, Technische Universit\"at Dresden, D-01062, Dresden, Germany}
\affiliation{Turku Center for Quantum Physics, Department of Physics and Astronomy,
University of Turku, FI-20014, Turun Yliopisto, Finland}
\author{Walter T. Strunz}
\affiliation{Institut f\"ur Theoretische Physik, Technische Universit\"at Dresden, D-01062, Dresden, Germany}

\title{Joint measurability in non-equilibrium quantum thermodynamics}
\date{\today}

\begin{abstract}
{In this Letter we investigate  {the concept of} quantum work  {and its measurability} from the viewpoint of quantum measurement theory.}
{Very often}, quantum work and fluctuation theorems are discussed in the framework of projective two-point measurement (TPM) schemes. According to a well known no-go theorem, there is no work observable which satisfies both (\ref{average_cond}) an average work condition and (\ref{eq:fluctuation-condition}) the TPM statistics for diagonal input states.
 {Such} projective measurements  {represent a restrictive class among all possible measurements.} It is desirable,  {both from a theoretical and experimental point of view, to extend the scheme to the general case including suitably designed unsharp measurements. This shifts the focus to the question what information about work and its fluctuations one is able to extract from such generalized measurements.} 
 {We show that} the no-go theorem no longer holds if the observables in a TPM scheme are jointly measurable for any intermediate unitary evolution.
We explicitly construct  {a model} with  {unsharp} energy measurements and derive bounds for the  {visibility} that ensure joint measurability. In such an  {unsharp} scenario a single work measurement apparatus can be  {constructed}  {that allows us to} determine the correct average work and to obtain free energy differences with the help of a Jarzynski equality.   
\end{abstract}

\maketitle

\paragraph{Introduction ---} 
{A suitable definition of work in quantum systems has proven to be elusive as witnessed
by a vast amount of different proposals in the last two decades} \cite{chernyak_2004,TalknerFluctuationtheoremsWork2007,huberEmployingTrappedCold2008,EspositoNonequilibriumfluctuationsfluctuation2009,CampisiColloquiumQuantumfluctuation2011,DornerExtractingQuantumWork2013,RoncagliaWorkMeasurementGeneralized2014,salmilehtoQuantumDrivingWork2014,DeffnerQuantumworkthermodynamic2016,AbergFullyQuantumFluctuation2018,LostaglioQuantumFluctuationTheorems2018,dongFunctionalFieldIntegral2019,feiGrouptheoreticalApproachCalculation2019,ortegaWorkDistributionsQuantum2019,micadeiQuantumFluctuationTheorems2020,Beyer2020,silvaQuantumMechanicalWork2021,gherardiniEndpointMeasurementApproach2021}.
Among them the two-point measurement (TPM) scheme \cite{tasakiJarzynskiRelationsQuantum2000,kurchanQuantumFluctuationTheorem2001,TalknerFluctuationtheoremsWork2007,guarnieriQuantumWorkStatistics2019,talknerColloquiumStatisticalMechanics2020} has become one of the standard approaches. Its popularity is certainly related to its ability to directly recover fluctuation theorems known from classical non-equilibrium thermodynamics, such as the ones by Jarzynski and Crooks~\cite{DeffnerNonequilibriumworkdistribution2008,albashFluctuationTheoremsQuantum2013,RoncagliaWorkMeasurementGeneralized2014,JarzynskiQuantumClassicalCorrespondencePrinciple2015,Gooldrolequantuminformation2016,millerTimereversalSymmetricWork2017,HayashiMeasurementbasedformulationquantum2017,zaninExperimentalQuantumThermodynamics2019,ribeiroExperimentalStudyGeneralized2020}. 
Despite the formal agreement of the quantum fluctuation theorems with their classical counterparts, several conceptual differences {remain}. Most prominently, the work obtained by two projective quantum energy measurements is generally not consistent with the average energy change in a closed system~\cite{AllahverdyanFluctuationsworkquantum2005,TalknerFluctuationtheoremsWork2007,AllahverdyanNonequilibriumquantumfluctuations2014}. {Rather recently, it was rigorously} shown that this discrepancy is fundamental and cannot be solved by a more clever definition of fluctuating work: The no-go theorem {of} Ref.~\cite{Perarnau-LlobetNoGoTheoremCharacterization2017}  {(see also Ref.~\cite{hovhannisyanEnergyConservationJarzynski2021})} states that there is no quantum work observable which satisfies the work fluctuation statistics of the projective TPM scheme and  {also} reproduces the average change of energy.

A projective TPM scenario consists of the following consecutive operations.
First, a closed quantum system is initialized in a state $\rho$.  {Then an energy measurement with respect to an initial Hamiltonian $\hat H_A=\sum_a E_a\,\Pi_a$ is performed and an outcome $a$ is recorded. A unitary process $U_{A\rightarrow B}$ generated by a time dependent system Hamiltonian is applied to the output state 
and a second energy measurement with respect to the final Hamiltonian $\hat H_B = \sum_b E'_b\,\Pi'_b$ is implemented, yielding outcome $b$}.
 {The energy change in the closed system is interpreted as work $w(a,b) = E'_b-E_a$. If we are only interested in the work values $w$ but not in the initial and final energies $E_a$ and $E'_b$ themselves, we may
ask if the TPM protocol can be replaced by a single quantum measurement $\obsW$ which
directly yields $w$.}

 {
Work should be defined for any process but its value is process dependent.
Thus, we look for an observable $\obsW$ that exists for any process $U_{A\rightarrow B} = U$ but may depend on it.}
{The no-go theorem is a result
of the clash of physically well motivated constraints set for such a work observable $\obsW$.
 {(I)} Averaging over all outcomes $w$ should reproduce the average energy change: $\Tr[\obsW\rho]=\Tr[\hat H_B \, U\rho U^\dagger]-\Tr[\hat H_A\,\rho]$.  {(II)} For initial states $\rho_{\Delta}$, {diagonal in the basis of the first energy measurement $\{\Pi_a\}$}, the statistics of $\obsW$ should agree with the statistics of the sequential measurement of $\hat H_A$ and 
$\hat H_B$: $p_{\obsW}(w=E'_b-E_a)=p_{\mathrm{TPM}}(a,b)$.  {Here one presumes the non-degeneracy of $w$, which is necessary to  prove the no-go theorem~\cite{Perarnau-LlobetNoGoTheoremCharacterization2017}. Later we will replace this assumption by a more general one.}}

{The above constraints are motivated by classical thermodynamics. The  {first} constraint   {(I)} states that the average work is equal to the average change of internal energy in a closed system. 
If the initial state is a Gibbs state, the
 {second} constraint  {(II)} ensures that the work satisfies Jarzynski's fluctuation theorem well known from classical stochastic thermodynamics.}

Projective energy measurements are an idealization.  {In the past it has often proven insightful to depart from this severe limitation and to consider generalized measurements described by a \textit{positive operator valued measure} (POVM), which can be seen as a projective measurement on an enlarged Hilbert space.
Most prominently and relevant for this paper, the compatibility of two POVMs is no longer determined by their commutativity as we know it for projective measurements~\cite{busch_unsharp_1986}.  In certain scenarios POVMs perform better than projective measurements. For instance, the optimal measurement for quantum state tomography is a POVM~\cite{RenesSymmetricInformationallyComplete2004} and 
POVMs can lead to stronger violations of some Bell inequalities~\cite{vertesiTwoqubitBellInequality2010}.
POVMs are also of experimental relevance. It has been shown that an ideal projective measurement would require infinite resources~\cite{guryanovaIdealProjectiveMeasurements2020}. Real experiments are often indirect and the actual measurement is indeed better described by a POVM with unsharp elements.}

 {Using POVMs, the question arises of how to estimate the desired quantity from the measurement outcomes.}
In the projective case it is clear that a given detector click can be associated with an eigenvalue of the Hamiltonian.  {Once we deviate from the projective scenario we have to carefully analyze what we can actually learn from the measurement outcomes. The latter are the only quantities we}  {can access in an experiment.}

TPM schemes beyond the projective case have been proposed, for example, in Ref.~\cite{prasannavenkateshTransientQuantumFluctuation2014,RasteginQuantumFluctuationsRelations2018,itoGeneralizedEnergyMeasurements2019,debarbaWorkEstimationWork2019}. 
There, the authors study different generalized energy measurements, based, for instance, on Gaussian pointers, and analyze their thermodynamic implications, in particular with respect to fluctuation theorems.  
In the present work we add an aspect to the study of generalized TPM schemes which, to the best of our knowledge, has not been addressed until now: Are given generalized energy measurements able to  {provide the correct average work and satisfy a fluctuation theorem?} 
What are the fundamental limitations for a generalized TPM scheme to allow for a ({unsharp}) work observable $\obsW$? This question is not only of theoretical interest but also experimentally relevant, since non-projective measurements are often easier to implement, especially if the measurement has to be non-destructive, as it is the case for the first measurement in a TPM scheme.
Instead of projective energy measurements given by the Hamiltonians $\hat H_{A,B}$, we consider two general POVMs $\A$ and $\B$ at the beginning and at the end of the protocol, respectively. $\A$ and $\B$ can be thought of as unsharp versions of the projective energy measurements. However, we keep them fully general for the moment and only later investigate their properties for a  {widely used class of unsharp POVMs.} 

{In a projective measurement each projector can unambiguously be associated with a corresponding eigenvalue $\Pi_a \rightarrow E_a$. Crucially, such a correspondence does not exist for POVMs, as the latter are not given by the spectral decomposition of a Hermitian operator. The energy associated with a certain POVM element $\A_a$ can in general only be an estimate based on additional assumptions. To reflect and emphasize this ambiguity, we introduce energy assignment functions that yield an energy estimate for the specific outcomes of the POVM $\A_a \rightarrow f(a)$ and $\B_b \rightarrow g(b)$. Only in the limiting case of projective energy measurements their ambiguity vanishes and $f(a) = E_a$ and $g(b) = E'_b$.}

 {Refining the constraints  {(I) and (II)} for the case of general observables, we find that the existence of $\obsW$ is closely related to the concept of joint measurability~\cite{busch_unsharp_1986,heinosaari_mathematical_2011,uola_adaptive_2016,skrzypczykComplexityCompatibleMeasurements2020}. We show that the no-go theorem can be extended to all energy POVMs $\A$ and $\B$ that are incompatible for at least some intermediate evolution $U$. In turn, if $\A$ and $\B$ are jointly measurable for any process $U$, the no-go statement generally does not apply. We illustrate that with a specific model for unsharp energy measurements.}

\begin{figure}[ht]
  \centering
  \includegraphics[width=.8\columnwidth]{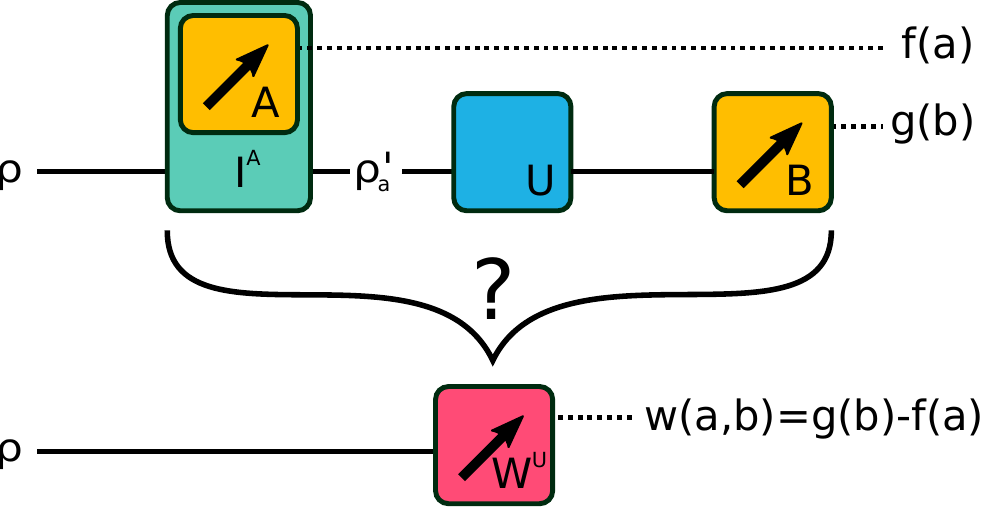}
  \caption{In a general TPM scheme a system is prepared in a state $\rho$ and measured by an initial measurement $\A$. This first measurement is implemented by an instrument $\Inst^\A$ that outputs a post measurement state $\rho_a'$ depending on the outcome $a$. After an arbitrary unitary evolution $U$ a final measurement $\B$ is implemented and outcome $b$ is obtained. For each outcome $a$ and $b$, an energy  {estimate} $f(a)$ and $g(b)$ is assigned and a work value $w(a,b)=g(b) - f(a)$ can be determined. In the case of projective measurements the  {estimates} $f(a)$ and $g(b)$ can uniquely be identified with the energy eigenvalues $E_a$ and $E'_b$. For general POVMs such a correspondence is missing.  
  We investigate whether such a scheme can be replaced by a single observable $\obsW^U$ that (\ref{average_cond}) reproduces the correct average work for any input state as well as (\ref{eq:fluctuation-condition}) the full statistics for input states diagonal in the initial energy eigenbasis of the system  {(see main text for precise formulations)}. For the projective case it is known that such an observable cannot exist. In contrast, if $\A$ and $\B$ are   {jointly measurable} POVMs for any $U$, the answer can be positive.}
  \label{fig:scheme}
\end{figure}

\vspace{1em}
\paragraph{{Definitions ---}}
We restrict ourselves to measurements with discrete, i.e., countably many outcomes. 
A discrete POVM is described by a set of \textit{effects} $\A = \{\A_a\}$, with positive elements $\A_a \geq 0$ and $\sum_a \A_a = \id$. 
The probability to obtain outcome $a$ when measuring a system in state $\rho$ is given by Born's rule: $p(a) = \Tr[\A_a\,\rho]$.

As motivated above the energy estimate for a specific outcome $a$ is given by a function $f(a)$. We can then define a generalized (unsharp) energy observable $\obsA$ as a tuple $(\A,f)$.  
Two observables $\mathcal A$ and $\mathcal{A}'$ that only differ in the function $f$ can be measured with the same apparatus as they share a common set of effects $\A$.
The average operator of $\obsA$ is defined as $\hat A_f = \sum_a f(a)\,\A_a$.

Two {POVMs $\A$ and $\B$} with $m$ and $n$ outcomes, respectively, are said to be \textit{jointly measurable} if they are the marginals of a third {POVM} $\W$ with $m \times n$ outcomes, i.e., $\A_a = \sum_b \W_{ab}$, $\B_b = \sum_a \W_{ab}$~\cite{busch2016quantum}. Two measurements that are not jointly measurable are \textit{incompatible}.
For projective measurements, joint measurability reduces to commutativity~\cite{heinosaari_mathematical_2011,heinosaari_invitation_2016,uola_adaptive_2016}.

The effects of $\A$ only determine the probabilities of the measurement outcomes. In a TPM scheme we also have to specify the post-measurement state of measurement $\A$. This can be done with the concept of \textit{quantum instruments}~\cite{davies1976quantum,heinosaari_mathematical_2011}.
An $\A$-compatible instrument $\Inst^{\A}$ is a set of completely positive maps $\{\Inst^\A_a\}$ that satisfy $\Tr[\Inst^\A_a(\rho)] = \Tr[\A_a \, \rho],\,\forall \rho$.
Any $\A$-compatible instrument can be written as $	\Inst^\A_a(\rho) = \mathcal{E}_a [\A_a^{{1}/{2}} \rho \, \A_a^{{1}/{2}}],$
where $\mathcal{E}_a$ is a completely positive trace-preserving (CPT) map which may depend on the outcome of the measurement. 
If $\mathcal{E}_a = \id,\,\forall a$, the instrument is called the \textit{Lüders instrument} of $\A$~\cite{heinosaari_mathematical_2011,busch2016quantum}. 
{Note that in the projective TPM scheme the Lüders instrument is usually tacitly assumed.}

\vspace{1em}
\paragraph{{A general TPM scheme ---}}
A \textit{general TPM} (GTPM) scheme is depicted in Fig.~\ref{fig:scheme}. An observable \mbox{$\obsA =(\A, f)$}, whose effects $\A$ are implemented by an instrument $\Inst^\A$, is measured on a system initialized in state $\rho$. After a unitary evolution $U$, a second observable \mbox{$\obsB=(\B, g)$} is measured. The probability to obtain an outcome sequence $(a,b)$ is given by  
\begin{align}
    p_\textrm{GTPM}(a,b) = \Tr[\B_b\, U \, \Inst_a^\A(\rho) \, U^\dagger] = \Tr[\B_b^U\, \Inst_a^\A(\rho)],
\end{align} 
where we have absorbed the intermediate evolution to the second measurement with the Heisenberg {picture} operators $ \B^U_b = U^\dagger \, \B_b \, U$.

\vspace{1em}
\paragraph{{Work observable $\obsW$} ---}
Can the GTPM scheme be replaced by a generalized work measurement? We consider an observable of the form \mbox{$\obsW^U = (\W^U,w(a,b) = g(b) - f(a))$}, where the superscript $U$ indicates the dependence on the intermediate evolution. 
The following two conditions shall apply.
\begin{enumerate}[(i)] 
    \item\label{average_cond} \textit{Average condition}: The average operator of $\obsW^U$ agrees with the difference of the average operators of $\obsB^U$ and $\obsA$ for all choices of $f$ and $g$.
    \begin{align*}
        \hat W^U_w = \hat B^U_g - \hat A_f, \quad \forall f,\,g.
    \end{align*}
    
    \item\label{eq:fluctuation-condition} \textit{Fluctuation condition:}
    For states $\rho_\Delta$ diagonal in {the initial energy eigenbasis} the statistics given by $\W^U$ agrees with the sequentially measured GTPM statistics:
    \begin{align*}
        \Tr[\W^U_{a b}\, \rho_\Delta] = \Tr[\B_b^U\,\Inst^\A_a(\rho_\Delta)].
    \end{align*}
\end{enumerate}
We note that in the first condition, the subscripts refer to energy assignment functions, whereas the subscripts in the second condition refer to measurement outcomes.  {These conditions can be seen as generalizations of the conditions (I) and (II) in the projective case.}

\begin{lemma*}\label{lemma:1}
An observable $\obsW^U$ satisfying the average condition~(\ref{average_cond}) can be found for any choice of $f$ and $g$ and any intermediate evolution $U$ if and only if $\A$ and $\B^U$ are jointly measurable for any $U$.
\end{lemma*}
\begin{proof}
For the if-direction we assume that $\W^U$ is a joint POVM of $\A$ and $\B^U$. 
\begin{align}
    \hat W^U_w &= \sum_{a b} \W^U_{a b} w(a,b) 
    =\sum_b \B^U_b g(b) - \sum_a \A_a f(a) \notag \\ 
    &=  \hat B^U_g - \hat A_f.   
\end{align}
For the only-if-direction we fix an outcome $a'$ and set $f(a') = 1$, $f(a \neq a') = g(b) = 0$. Plugging this choice into the average condition we get $\sum_{b} \W_{a' b}^U = \A_{a'}$.
Applying the same reasoning to all POVM elements $\B_b^U$ and $\A_a$, we can conclude that
\begin{align}
    \sum_a \W_{ab}^U = \B_b^U &&  \sum_b \W_{ab}^U = \A_a.
\end{align}
Thus, $\W^U$ is a joint POVM of $\A$ and $\B^U$.
\end{proof}
\begin{corollary*}\label{corollary:1}
A work observable does \textit{not} exist if the measurements $\A$ and $\B^U$ are incompatible for some unitary $U$. A no-go result for projective measurements immediately follows from this condition since joint measurability is equivalent to commutativity in the projective case and two non-trivial projective observables cannot commute for every $U$.
\end{corollary*}

The fluctuation condition (\ref{eq:fluctuation-condition}) depends --- contrary to condition (\ref{average_cond}) --- on the choice of the instrument $\Inst^\A$ implementing the first measurement and, thus, further restricts the form of the joint POVM $\W^U$. However, in the following   {paragraph} we show first that such  {an unsharp} observable can indeed exist in a  {non-trivial and} physically well motivated quantum thermodynamical scenario. Then we investigate the role of the  {energy assignments} $f,g$ and how they could be chosen based on different physical arguments.

\vspace{1em}
\paragraph{A  {model for a} joint observable of   {unsharp energies} ---}
We construct an  {unsharp} observable $\obsW^U$ and show that both conditions~(\ref{average_cond}) and~(\ref{eq:fluctuation-condition}) can be satisfied. 
Let $\hat H_A = \sum_{a=1}^d E_a \Pi_a$ and $\hat H_B=\sum_{b=1}^d E'_b \Pi'_b$ be two Hamiltonians of a $d$-dimensional quantum system.
 {Let us consider an apparatus which measures a mixture of the projectors instead of the perfectly sharp elements themselves. Thus, the effects have the form
\begin{align}
    \A_a = \sum_i R_{ai} \, \Pi_i, && \B_b = \sum_j S_{bj}\, \Pi'_j,
\end{align}
where $R$ and $S$ are bistochastic matrices.
The simplest one-parameter model of this form, often considered in the literature, is given by symmetric matrices $R_{a=i} = \lambda + (1-\lambda)/d$, $R_{a \neq i} = (1-\lambda)/d$ and $S_{b=j} = \gamma + (1-\gamma)/d$, $S_{b \neq j} = (1-\gamma)/d$~\cite{busch_unsharp_1986,uola_adaptive_2016,designolleIncompatibilityRobustnessQuantum2019}. (This choice is basically equivalent to a depolarizing noise which is omnipresent in many experimental setups~\cite{vovroshSimpleMitigationGlobal2021,urbanekMitigatingDepolarizingNoise2021}.)
We will stick to this model here, since it allows for a fully analytical treatment for any dimension $d$.}
{Thus, the}  {unsharp effects are given by the original projectors mixed with the identity}  {(sometimes referred to as adding white noise).}
\begin{align}
	\label{eq:noisy-POVMs}
	\A_a = \lambda \Pi_a + \frac{1-\lambda}{d} \id, &&
	\B_b = \gamma\, \Pi'_b + \frac{1-\gamma}{d} \id ,
\end{align}
where $0 < \lambda < 1$ and $0 < \gamma < 1$ are $\textit{visibility}$ parameters interpolating between the sharp energy measurements ($\lambda, \gamma = 1$) and the trivial POVM ($\lambda, \gamma = 0$). The first measurement is implemented by a Lüders instrument, i.e., $\Inst^\A_{a}(\rho) = \A_a^{1/2} \, \rho \, \A_a^{1/2}$. 
As an ansatz for the joint POVM $\W^U$ we propose the following form for its effects:
\begin{align}
	\label{eq:joint-effects}
	\W^U_{a b} = \A_a^{1/2} \, \C_b \, \A_a^{1/2},
\end{align}
where $\C$ is an auxiliary POVM which is constructed in the following. From the average condition (\ref{average_cond}) we know that the marginals of $\W^U_{a b}$ must agree with the POVMs $\A$ and $\B^U$, so we have $\B_b^U = \sum_a \A_a^{1/2} \, \C_b \, \A_a^{1/2}  = \mathcal{I}_\A(\C_b)$
and, therefore,
\begin{align}
	\label{eq:C}
	\C_b  = \mathcal{I}^{-1}_\A(B_b^U) = \mathcal{I}^{-1}_\A(\Psi_\gamma(U^\dagger\, \Pi_b' \, U)),
\end{align}
with the  {depolarizing} channel $\Psi_\gamma(X) = \gamma X + (1-\gamma)/d\, \id$.

Since measurement $\A$ is implemented by a Lüders instrument, the map $\mathcal{I}_\mathsf{A}$ is unital. The inverse map (which exists for any $\lambda,\gamma\neq 1$) is unital, too, and, thus, also the marginal condition for $\A$ is satisfied: \mbox{$\sum_b \A_a^{1/2} \, \mathcal{I}^{-1}_\A  (\B_b^U)\, \A_a^{1/2} = \A_a$}.

One can verify that the operators defined by Eqns.~(\ref{eq:joint-effects},\ref{eq:C}) fulfil the fluctuation condition (\ref{eq:fluctuation-condition}) for any input state diagonal in the initial energy basis $\{\Pi_a\}$ (see Supplemental Material~\footnote{See Supplemental Material at [URL will be inserted by
publisher] for details about the fluctuation condition and the derivation of the visibility bounds, which includes Ref.~\cite{bertlmannBlochVectorsQudits2008}}). Since the average condition (\ref{average_cond}) is satisfied by construction, it only remains to establish bounds for the visibility parameters $\lambda$ and $\gamma$ that guarantee positivity of $\W^U$ for any intermediate evolution $U$.  

It is well known that  {any} two  {given} POVMs $\A$ and $\B$  {become} jointly measurable if enough white noise is added  {as in Eq.~(\ref{eq:noisy-POVMs})}~\cite{uola_adaptive_2016}. In our case the situation is slightly different because we require joint measurability of $\A$ with the Heisenberg $\B^U$ for \textit{arbitrary} $U$. Thus, while the  {unsharp} energy measurements $\A$ and $\B$ are fixed, the joint POVM $\W^U$ has to depend on the evolution. Therefore, we are looking for parameters $\lambda$ and $\gamma$ such that $\W_{a b}^U \geq 0, \forall a,b,U$.

By construction, $\W_{a b}^U$ is positive if $\C_b$ is positive. Thus, from Eq.~(\ref{eq:C}) we see that the concatenation of the depolarizing channel with the inverse instrument $\Xi = \Inst^{-1}_\A \circ \Psi_\gamma$ needs to be a positive (but not necessarily completely positive) map. 
A map $\Xi$ is positive if its Choi state has a non-negative expectation value for any product state~\cite{zyczkowskiDualityQuantumMaps2004a}. 
In our case, this condition can be checked analytically for any dimension $d$ (see Supplemental Material~\cite{Note1}). The map $\Xi$ is positive, and therefore $\W^U$ is a valid POVM if the following relation for the two visibility parameters $\lambda$ and $\gamma$ holds:
\begin{align}
 	\label{eq:noise-condition}
    \gamma \leq \frac{2 \kappa}{d+2 \kappa -d \kappa},
\end{align}
with
\begin{align}
	\label{eq:kappa}
	\kappa = 2 \sqrt{\lambda + \frac{1-\lambda}{d}}\sqrt{\frac{1-\lambda}{d}} + (d-2)\frac{1-\lambda}{d}.
\end{align}
Thus, for $\lambda,\gamma$ satisfying Eq.~(\ref{eq:noise-condition}) the GTPM scheme can be replaced by a single work observable $\obsW^U$.

For $d>2$ the visibilities $\lambda$ and $\gamma$ are strictly smaller than the best known bounds for joint measurability of white noise affected measurements~\cite{designolleIncompatibilityRobustnessQuantum2019}. This is expected since condition (\ref{eq:fluctuation-condition}) imposes an additional constraint. The ansatz for the effects $\W^U_{a b}$ in Eq.~(\ref{eq:joint-effects}) is a priori not the most general one~\cite{heinosaariUniversalitySequentialQuantum2015}. However, numerical simulations suggest that our  {visibility} bounds are tight and we conjecture that our construction is optimal to satisfy condition (\ref{eq:fluctuation-condition}) for the given model Eq.~(\ref{eq:noisy-POVMs}) (see Supplemental Material for details~\cite{Note1}).

\vspace{1em}
\paragraph{Energy assignments ---}
Up to now we kept $f$ and $g$ arbitrary. 
To establish the connection to thermodynamic quantities we will now consider possible assignments of energies to the measurement outcomes. 
The visibility bounds derived above ensure that a measurement apparatus exists that can replace the GTPM scheme. The measurement outcomes of such an apparatus can then be used to determine different quantities of interest depending on the chosen assignments $f$ and $g$. For the following physically motivated choices we always require that the energy eigenvalues are obtained in the limit of  {maximal visibility} $f(a) \stackrel{\lambda \to 1}{=} E_a$, $ g(b) \stackrel{\gamma \to 1}{=} E'_b$.

\paragraph{Average work ---} The first obvious choice are the energy eigenvalues themselves  $f(a) = E_a$, $g(b)= E'_b$. The experimenter assigns energies as if the measurements were  {projective}. Clearly, such an assignment does not yield the correct energy expectation values of the Hamiltonians $\hat H_A$ and $\hat H_B$.  {Since the measured effects are unsharp, the determined energy is}
shifted towards the average energy of a fully mixed state $\overline{E}_{A,B} = \nicefrac{1}{d}\Tr[H_{A,B}]$. 
 {Thus, the experimenter compensates for this shift} by assigning corrected energies
\begin{align*}
    f(a) = \frac{1}{\lambda}\, E_a - \frac{1 - \lambda}{\lambda} \overline{E}_A, &&
    g(b) = \frac{1}{\gamma} \, E'_b - \frac{1 - \gamma}{\gamma}\overline{E}_B.
\end{align*}
This choice leads to the correct average energies of the marginal observables. Condition (\ref{average_cond}) then ensures that the average of the joint observable $\obsW^U$ agrees with the average work $\hat W^U_{w} = \hat H_B - \hat H_A = \hat W_\textrm{avg}$. Thus, even though the measurements are  {unsharp}, by  {correctly} choosing $f$ and $g$ the work observable $\obsW$ yields  {the} accurate value for the average work in the system  {for any input state}.  {This might seem surprising at first glance. However, the unsharp POVMs are informationally equivalent to the projective energy measurements in the sense that the former can be used to reconstruct the statistics of the latter.}

\paragraph{Free energy difference ---}
We say a GTPM scheme satisfies a generalized Jarzynski equation (JE) if the following equation holds for any unitary $U$:
\begin{align}
    \label{eq:JE}
	\sum_{a,b} \tr[\B_b\, U \Inst^\A_a (\rho_0) U^\dagger] e^{- \beta (g(b) - f(a))} = e^{-\beta \Delta \tilde F},
\end{align}
where $\Delta \tilde F$ is real and independent of the unitary $U$ on the left hand side. This can in general only be satisfied for any $U$ if we demand
\begin{align}
	\label{eq:summing-to-identity}
	\sum_a e^{+\beta f(a)} \Inst^\A_a[\rho_0] =  \alpha \id,
\end{align}
where $\alpha$ is a positive real number~\cite{rasteginNonequilibriumEqualitiesUnital2013}. Here we only consider $\rho_0 = \rho_\textrm{Gibbs}$. In the limit of projective Lüders measurements ($\A_a^{1/2} = \Pi_a$), the standard JE is immediately recovered if we set $\alpha = 1/Z_A$, where $Z_A = \sum_a \exp[-\beta E_a]$, and we get $\Delta \tilde F = \Delta F$ on the right hand side of Eq.~(\ref{eq:JE}). Thus, as is well known, a projective TPM scheme can be used to determine the free energy difference $\Delta F$~\cite{TalknerFluctuationtheoremsWork2007}. However, the projective measurements cannot be replaced by a work observable $\obsW^U$. Accordingly, the correct average energies cannot  {be obtained from the same experimental setup.} 

Interestingly, Eq.~(\ref{eq:summing-to-identity}) can also be satisfied by   {unsharp} measurements if a suitable assignment $f$ is chosen:
\begin{align}
\label{eq:fas}
	f(a) = \frac{1}{\beta} \ln\left [ \frac{\alpha \, Z_A}{\lambda} \left( e^{+\beta E_a} - \frac{1-\lambda}{ d} \sum_i e^{+\beta E_i} \right) \right].
\end{align}
We can set again $\alpha = 1/Z_A$. Demanding $\Delta \tilde F = \Delta F$, we then find the assignment $g(b) = E_b'$. 

Therefore, if we choose visibilities satisfying Eq.~(\ref{eq:noise-condition}), condition (\ref{eq:fluctuation-condition}) guarantees that a joint observable $\obsW^U$ exists such that
\begin{align}
    \Tr\left[ \sum_{a,b} e^{-\beta w(a,b)} \W_{a b}^U\, \rho_\textrm{Gibbs} \right] = \frac{Z_B}{Z_A} = e^{-\beta \Delta F}.
\end{align}
Accordingly, we can use the POVM $\W^U$ to also obtain the free energy difference $\Delta F$ between the initial and final Hamiltonian. 

{It is worth noting that the GTPM scheme, unlike the projective one, generally yields fluctuations also in the trivial case of $\hat H_A=\hat H_B$ and $U = \id$, due to the inherent non-repeatability of unsharp measurements~\cite{buschRepeatableMeasurementsQuantum1995,heinosaari_mathematical_2011}. Importantly, the vanishing average work and free energy difference is, of course, correctly obtained also for this trivial scenario.}

\paragraph{Conclusion ---}
{The idealized projective TPM scenario rules out the existence of a work observable $\obsW^U$ that satisfies both an average and a fluctuation condition. Our results show that this is no longer true for the general scenario using unsharp measurements. The answer to the question of whether a work observable can exist is richer in this case and demands the consideration of joint measurability.}
{The} no-go theorem from Ref.~\cite{Perarnau-LlobetNoGoTheoremCharacterization2017} {can be extended} to the case of general energy POVMs {if the} two measurements $\A$ and $\B$ are incompatible for some intermediate unitary evolution $U$. 
However, if $\A$ and $\B^U$ are jointly measurable for every $U$, the no-go theorem can be overcome and a work observable $\obsW^U$ may exist. By explicit construction of $\obsW^U$ for a specific model  {of unsharp measurements} we have shown that a general TPM scheme can indeed be replaced by a single measurement in a physically well motivated scenario.

For non-projective measurements a unique correspondence between POVM elements and associated energies is missing. Energy assignment functions $f(a)$ and $g(b)$  {have to be specified.} The existence of $\obsW^U$ is independent of that choice. 
 {However, $f$ and $g$ are determined by the quantity one would like to estimate from the measurement outcomes and depend on the effects of the POVM.}
For the specific  {model} we show that a single measurement apparatus implementing the POVM $\W^U$ can be used to determine both the correct average work $W_\textrm{avg}$ for any input state and the correct free energy difference $\Delta F$ between the Hamiltonians $\hat H_A$ and $\hat H_B$ if the system starts in a Gibbs state.

Theoretical investigations on quantum work and fluctuations  {should embrace the rich possibilities offered by quantum measurement theory.}  
 {Our work highlights the central role of joint measurability for TPM schemes.}
 {We show yet again that unsharp measurements enable to achieve goals that are unreachable with projective ones, here with important applications in quantum thermodynamics.}
Joint measurability offers an operationally motivated approach to the quantum-to-classical transition and we have shown that this transition is closely related to the existence of a work observable.  {In the future} it will be interesting to explore how other operational characterizations of quantumness, such as inherent measurement disturbance, emerge in thermodynamic settings.

\clearpage
\newpage
\section{Supplemental Material}

\renewcommand{\theequation}{S.\arabic{equation}}
\setcounter{equation}{0}

\section{The observable fulfils the fluctuation condition} 
 
To see that also the fluctuation condition (ii) is satisfied by the proposed joint observable, we have to show that it preserves the TPM statistics with respect to input states diagonal in the first energy basis
\begin{align}
	\label{eq:diagonal-state}
	\rho_0 = \sum_k p_k \, \Pi_k,
\end{align}
with $p_k > 0$ and $\sum_k p_k =1$.
All POVM elements of the first measurement $\A_m$ are also diagonal in the first energy basis and can be written as
\begin{align}
	\A_m = \sum_i \alpha_{m,i} \, \Pi_i.
\end{align} 
The measurement statistics is then given by
\begin{align}
	p&_\textrm{GTPM}(a,b) = \Tr[\B_b^U\, \A_a^{1/2}\, \rho_0\, \A_a^{1/2}  ]\notag\\
	&= \Tr\left[\left(\sum_m \A_m^{1/2} \, \C_b\, \A_m^{1/2}   \right) \, \A_a^{1/2}\, \rho_0\, \A_a^{1/2} \right]\notag\\
	&=\Tr \left[\left(\sum_{m,i,j} \sqrt{\alpha_{m,i}} \, \Pi_i\, \C_b\, \Pi_j\sqrt{\alpha_{m,j}}    \right) \, \right.\notag\\ &\hspace{4em} \left.\times \sum_{p,k,q} \sqrt{\alpha_{a,p}} \, \Pi_p \, p_k\, \Pi_k \, \Pi_q  \sqrt{\alpha_{a,q}} \right] \notag\\
	&=\Tr \left[\sum_{k,m} \alpha_{m,k}\, \Pi_k\, \C_b \, \Pi_k \, \alpha_{a,k}\, p_k \right] \notag\\
	&=\Tr \left[\sum_{k} \Pi_k\, \C_b \, \Pi_k \, \alpha_{a,k}\, p_k \right] \notag \\
	&=\Tr \left[\C_b  \, \sum_{p,k,q}\sqrt{\alpha_{a,p}}\,\Pi_p\,   \Pi_k\, \Pi_q \, \sqrt{\alpha_{a,q}} \,p_k \right]\notag\\
	&=\Tr\left[\A_a^{1/2}\, \C_b \, \A_a^{1/2}\, \rho_0  \right] = \Tr[\W_{a b}^U\, \rho_0]\notag \\
	&=p_\W(a,b).
\end{align}
Accordingly, the fluctuation condition is satisfied. 

\section{Derivation of the visibility bounds}
Here we explicitly derive the relation between $\gamma$ and $\lambda$ that guarantees positivity of the effects $\C$ and, therefore, the existence of a joint POVM $\W^U$. We have $\C_b = \Inst^{-1}_\A \circ \Psi_\gamma [U^\dagger \, \Pi_b' \, U]$. Thus, the map $\Xi = \Inst^{-1}_\A \circ \Psi_\gamma$ has to be positive (but not necessarily completely positive).  

First of all, we calculate the inverse instrument channel $\Inst_{\A}^{-1}$. To this end we write the map $\Inst_{\A}$ in a generalized Bloch vector representation. Any density operator $\rho$ of a $d$-dimensional quantum system can be decomposed as 
\begin{align}
	\label{eq:Bloch-representation}
	\rho &= \sum_{\mu=0}^{d^2-1}  \Tr[G_\mu\, \rho]\, G_\mu = \sum_{\mu=0}^{d^2-1} \Gamma_\mu \, G_\mu \notag\\
	&= r_0 G_0 + \sum_{\mu=1}^{d^2-1} \mathbf{r}_\mu \, G_\mu ,
\end{align}
where $G_0 = \id/\sqrt{d}$ and the $\{G_1,\ldots,G_{d^2-1}\}$ are the Hermitian and traceless generators of the group $SU(d)$ with $\Tr[G_\mu\,G_\nu] = \delta_{\mu\nu}$. The vector $\mathbf{r}$ is called the Bloch vector of $\rho$.

Fixing an orthonormal basis $\{\ket{a}\}$ of the $d$-dimensional Hilbert space, the $G_\mu$ are explicitly given by~\cite{bertlmannBlochVectorsQudits2008}
\begin{align}
\label{eq:Gell-Mann}
	&\{G_\mu\}_{\mu=1}^{\frac{d}{2}(d-1)} = \frac{1}{\sqrt{2}}(\kb{j}{k}+\kb{k}{j}),\\
	&\{G_\mu\}_{\mu=\frac{d}{2}(d-1)+1}^{d^2-d} = -\frac{i}{\sqrt{2}}(\kb{j}{k}-\kb{k}{j}),\\
	&\{G_\mu\}_{\mu=d^2-d}^{d^2-1} =\sqrt{\frac{1}{l(l+1)}}\left(\sum_{n=1}^l\kb{n}{n}-l\kb{l+1}{l+1}\right),
\end{align} 
where $1\leq j < k \leq d$ and $1\leq l \leq d-1$.
A matrix representation of the instrument channel $\Inst_{\A}$ is then given by
\begin{align}
	\label{eq:map-construction}
	\Lambda_{\mu\nu} = \Tr[G_\mu\, \Inst_{\A}(G_\nu)].
\end{align}
By choosing the basis $\{\ket{a}\}$ to be equal to the energy eigenbasis corresponding to the first measurement, i.e., $\kb{a}{a} = \Pi_a$, one can verify that the matrix $\Lambda_{\mu\nu}$ becomes diagonal with the entries
\begin{align}
	\Lambda_{\mu\mu} = 
	\begin{cases}
	1, \text{if}\ \mu = 0\\
		\kappa, \text{if}\ 1\leq \mu \leq d^2-d\\
		1, \text{otherwise}
	\end{cases} ,
\end{align}
where 
\begin{align}
	\label{eq:kappa_sup}
	\kappa = 2 \sqrt{\lambda + \frac{1-\lambda}{d}}\sqrt{\frac{1-\lambda}{d}} + (d-2)\frac{1-\lambda}{d}.
\end{align}
The depolarizing map $\Psi_\gamma$ is also diagonal in the Bloch representation and reads
\begin{align}
    \Upsilon_{\mu \mu} = 
    \begin{cases}
		1, \text{if}\ \mu = 0\\
		\gamma, \text{otherwise}
	\end{cases} ,
\end{align}
Thus, the full map is given by
\begin{align}
   \left(\Lambda^{-1}\, \Upsilon\right)_{\mu \mu} = 
    \begin{cases}
		1, \text{if}\ \mu = 0\\
		\frac{\gamma}{\kappa} , \text{if}\ 1\leq \mu \leq d^2-d\\
		\gamma, \text{otherwise}
	\end{cases}.
\end{align}
To check its positivity we calculate the corresponding Choi matrix 
\begin{align}
    D = d^2 \sum_{\mu=0}^{4 d^2 - 1} \left((\id \otimes \Lambda^{-1}\, \Upsilon )[\Gamma_{\Phi_+}]\right) G_\mu,
\end{align}
where $\Gamma_{\Phi_+}$ is the generalized Bloch vector of the Bell state $\kb{\Phi_+}{\Phi_+}$ and we have introduced a scaling prefactor $d^2$ for convenience.
Then $D$ has the form
\begin{align}
    D_{i j} =& (1-\gamma)\delta_{i j} + 3\gamma \sum_{n=1}^{d-1} \delta_{i j} \, \delta_{i, 1+ n(d+1)} \notag\\
    &+ \frac{3\gamma}{\kappa}\sum_{\substack{n,m = 1 \\ n \neq m}}^{d-1} \delta_{i,1+n(d+1)}\, \delta_{j,1+m(d+1)}. 
\end{align}
We define a generic product state
\begin{align}
    &\ket{\psi} = \ket{\psi_a} \otimes \ket{\psi_b}, \quad \textrm{with} \notag\\
    &\ket{\psi_a} = \sum_{i=1}^d a_i \ket{i}, \quad \ket{\psi_b} = \sum_{i=1}^d b_i \ket{i} ,
\end{align}
that we assume to be nomalized, i.e., \mbox{$\sum_i |a_i|^2 = \sum_i |b_i|^2 =1$}, w.o.l.g.
For the expectation value of $D$ we then find
\begin{align}
    \bra{\psi} D \ket{\psi} =& \sum_{i,j} |a_i|^2 |b_j|^2 + (d-1)\, \gamma \sum_i |a_i|^2 |b_i|^2 \notag \\
     & - \gamma \sum_{i\neq j} |a_i|^2 |b_j|^2 + \frac{d \, \gamma}{\kappa}\sum_{i\neq j} a_i b_i a_j^* b_j^* \notag \\
     =& 1 - \gamma + d\,\gamma \sum_i|a_i|^2 |b_i|^2 + \frac{d\,\gamma}{\kappa}\sum_{i\neq j} a_i b_i a_j^* b_j^* \notag \\
     =& 1 - \gamma + d\,\gamma \left( \sum_{i,j} |a_i|^2 |b_j|^2 - \sum_{i\neq j} |a_i|^2 |b_j|^2 \right) \notag \\
     & + \frac{d\,\gamma}{\kappa}\sum_{i < j} 2 \Re[a_i b_i a_j^* b_j^*] \notag\\
     =& 1 - \gamma + d\,\gamma - d\,\gamma \sum_{i\neq j} |a_i|^2 |b_j|^2 \notag\\
     &+ \frac{2 d\,\gamma}{\kappa}\sum_{i < j} \Re[a_i b_i a_j^* b_j^*]
 \end{align}
The last term is lower bounded by $\frac{d\,\gamma}{2 \kappa}$ and this minimum is attained if we choose, e.g., $a_1 = a_2 = b_1 =  1/\sqrt{2}$ and $b_2 = -1/\sqrt{2}$ while all other entries being zero. This choice also minimizes the second last term and we obtain the overall minimum
\begin{align}
    \min_{\ket{\psi} = \ket{\psi_a} \otimes \ket{\psi_b}} \bra{\psi} D \ket{\psi} = 1 - \gamma - \frac{d\,\gamma}{2} - \frac{d\,\gamma}{2 \kappa}. 
 \end{align}
Thus, if we require that the minimum must be non-negative for the map $\Xi$ to be positive, we get the desired condition for $\gamma$ and $\kappa$.
\begin{align}
    \gamma \leq \frac{2 \kappa}{d+2 \kappa -d \kappa},
\end{align}
with
\begin{align}
	\kappa = 2 \sqrt{\lambda + \frac{1-\lambda}{d}}\sqrt{\frac{1-\lambda}{d}} + (d-2)\frac{1-\lambda}{d}.
\end{align}

In Fig.~\ref{fig:lambda-comparison} we plot the bound for the visibility as a function of dimension $d$ for the symmetric case $\gamma = \lambda =\lambda_\textrm{sym}$. For comparison we also plot the corresponding values for the best known bound for a general joint POVM $\lambda_\textrm{opt}$ and the bound $\lambda_\textrm{mub}$ for measurements, where the underlying sets of noiseless projectors $\{\Pi_a\}$ and $\{\Pi'_b\}$ are mutually unbiased bases
\begin{align}
    \lambda_\textrm{opt} &= \frac{d-2 + \sqrt{d^2+4d -4}}{4(d-1)}, \\
    \lambda_\textrm{mub} &= \frac{1}{2}\left( \frac{1}{\sqrt{d}+1}\right).
\end{align}
The latter case does, of course, not apply for our scenario since the intermediate unitary $U$ will in general not lead to mutually unbiased bases. One can however wonder if the visibilities $\lambda_\textrm{opt}$ could be reached in our scenario. One has to keep in mind that the joint POVM in our scheme has to satisfy also the fluctuation condition (i). We ensure this by the special form of the joint POVM given in Eq.~(5) of the main text. The derivation of the general bound $\lambda_\textrm{opt}$ employs another form of the joint POVM which does not fulfil the fluctuation condition in general~\cite{designolleIncompatibilityRobustnessQuantum2019}.

\begin{figure}[ht]
  \centering
  \includegraphics[width=\columnwidth]{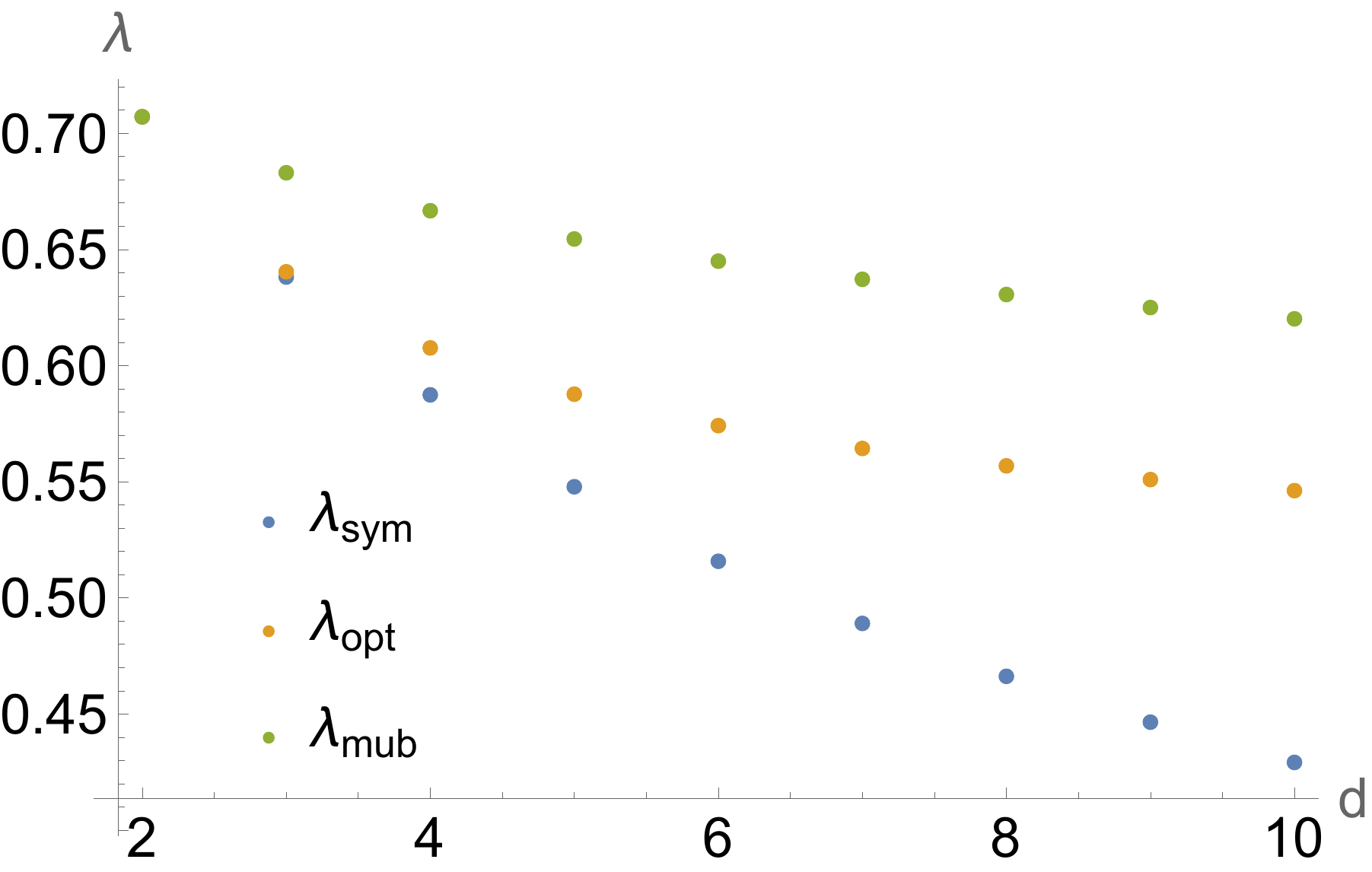}
  \caption{The plot shows the critical symmetric visibilities $\lambda = \gamma$ for dimensions $2\leq d \leq 10$. $\lambda_\textrm{opt}$ shows the best bound for our noise model known from the literature~\cite{designolleIncompatibilityRobustnessQuantum2019}. If one requires agreement also with condition (ii) we obtain lower visibilities $\lambda_\textrm{sym}$. We conjecture that this bound is tight. For comparison we also plot the bound $\lambda_\textrm{mub}$ known for mutually unbiased bases $\{\Pi_a\}$ and $\{U^\dagger \, \Pi'_b \, U\}$. As one can see, mutual unbiased measurements are not the most incompatible ones but allow for higher visibilities than the general case.}
  \label{fig:lambda-comparison}
\end{figure}

Based on numerical simulations we can only conjecture that our joint POVM $\W^U$ leads to the best possible visibility $\lambda_\textrm{sym} = \lambda = \gamma$ when condition (ii) is imposed. 
For a fixed visibility $\lambda_\textrm{sym}$ and a given unitary $U$ a suitable joint POVM $\W^U$ can be found through a semidefinite program (SDP). 
\begin{align}
    \min \sum_{a,b,k} &\left|\Tr[(\K_{a b} - \A_a^{1/2}\B_b^U \A_a^{1/2})\, \Pi_k]\right| \notag \\
    \textrm{s.t.:}& \notag \\
    &\sum_b \K_{a b} = \A_a  \\
    &\sum_a \K_{a b} = \B_b^U 
\end{align}
The first line minimizes the distance between the TPM statistics and the statistics given by the joint observable for diagonal input states. The second and third line ensure that the joint POVM $\K$ has the correct marginals.
If the problem is feasible a joint POVM $\W^U_{a b}=\K_{a b}$ exists for the chosen visibility $\lambda_\textrm{sym}$ and the unitary $U$. If the minimization reaches zero, the joint POVM also satisfies the fluctuation condition (ii). By randomly sampling unitaries $U$ and scanning through $0 < \lambda_\textrm{sym} < 1$ we can numerically reproduce our analytical visibility bounds. Therefore, we conjecture that the proposed form of the joint POVM is optimal for the given noise model if one wants to satisfy the fluctuation condition (ii). The Supplemental Material includes a basic \textsc{matlab} script which implements this procedure to find the best numerical bounds.

\bibliography{bibliography.bib}

\end{document}